\documentclass[10pt,final]{IEEEtran} 
\usepackage{fancyhdr}
\usepackage{epsfig}
\usepackage{threeparttable}
\usepackage{epsf,epsfig}
\usepackage{amsthm}
\usepackage{amsmath}
\usepackage{amssymb}
\usepackage{amsfonts}
\usepackage[noadjust]{cite}
\usepackage{dsfont}
\usepackage{subfigure}
\usepackage{color}

\newtheorem{proposition}{Proposition}

\def\phi{\varphi}

\def\l{\left}
\def\r{\right}
\def\({\left(}
\def\){\right)}

\def\b0{{\mathbf{0}}}

\newcommand{\Pout}{P_{\mathsf{out}}}

\newcommand{\nn}{\nonumber}

\begin{document}
\title{ Coverage and Economy of Cellular Networks with Many Base Stations}
\author{Seunghyun Lee and Kaibin Huang\thanks{ S. Lee and K. Huang are with Yonsei University, S. Korea. Email: lshnsy@gmail.com, huangkb@ieee.org.  Updated on \today }}

\maketitle
\begin{abstract}
The performance of a cellular network can be significantly improved by employing many base stations (BSs), which shortens transmission distances. However, there exist no known results on quantifying the performance gains from deploying many BSs. To address this issue, we adopt a stochastic-geometry model of the downlink cellular network  and analyze the mobile outage probability. Specifically, given Poisson distributed BSs, the outage probability is shown to diminish inversely with the increasing ratio between the BS and mobile densities. Furthermore, we analyze the optimal tradeoff between the performance gain from increasing the BS density and the resultant network cost accounting for energy consumption, BS hardware and backhaul cables. 
The optimal BS density is proved to be proportional to the square root of the mobile density and the inverse of the square root of the cost factors considered. 
\end{abstract}

\section{Introduction}
Compared with the advancement in physical-layer techniques, reducing the cell size by using more base stations (BSs) have resulted in much more significant throughput gains. This  observation has led to active research on the deployment additional BSs to shorten the transmission distances \cite{SurveyOn3GPPHetero:2011}. In view of prior work, it remains unclear what is the network performance in the limit of many BSs, which is addressed in  this letter. 

The difficulty of quantifying the asymptotic network performance lies in the lack of a practical and yet tractable cellular-network model. Traditionally, cellular networks are modeled using  the grid model  and the evaluation of their performance has  to rely on simulation (see e.g., \cite{OnTheCapacityCellularCDMA}). Recently,  a tractable stochastic-geometry model for a downlink cellular network has been proposed in \cite{Andrews:TractableApproachCoverageCellular:2010}, where BSs are modeled as  a homogeneous Poisson point process (PPP) which allows the derivation of the mobile outage probability in a relatively simple form. A key observation in \cite{Andrews:TractableApproachCoverageCellular:2010} is that the outage probability is insensitive to the change on the BS density assuming all BSs transmit.  This assumption does not hold  for  the scenario where BSs significantly outnumber mobiles and consequently many cells are empty. This scenario exists in heterogeneous networks where dense microcell and femtocell BSs are installed \cite{SurveyOn3GPPHetero:2011} or cellular networks where a large number of distributed antennas are deployed in each cell and each antenna functions as a virtual BS \cite{WanChoi:DistributedAntennas:2007}. In this letter,   the network model in \cite{Andrews:TractableApproachCoverageCellular:2010} is modified by preventing  BSs in empty cells from transmitting. In the limit of many BSs, the active-BS density converges to the mobile density while the transmission distances diminish. As a result, increasing the BS density can increase the received signal power  without causing additional inter-cell interference. The resultant network performance gain is quantified in this letter. Specifically, given fixed mobile density, it is shown that the outage probability diminishes inversely  with increasing BS density. It is verified by simulation that this result derived for asymptotically many BSs is accurate even in the practical range of BS density.      

Despite improvements on the network coverage, the deployment of many  BSs increases the network cost including the length of backhaul cables connecting BSs to switching centers \cite{Baccelli:StochGeometryArchitectCommNetwork:2006},  the BS hardware and the network energy consumption. In this letter, the BS density is optimized for achieving an optimal tradeoff between network performance and cost. To this end, the network model is augmented with an additional homogeneous PPP modeling switching centers. The optimal BS density is derived by minimizing a multi-objective cost function that accounts for the BS hardware, total cable length, total energy consumption and the outage probability. It is shown that the optimal BS density is proportional to the square root of the mobile density and inversely proportional to the square root of a linear combination of the network cost factors.

\section{Network Model} \label{Sec:SysMod}
BSs, mobiles  and switching centers are modeled as independent PPPs $\Sigma_b$, $\Sigma_u$ and $\Sigma_s$ of density $\lambda_b$, $\lambda_u$ and $\lambda_s$, respectively. Mobiles are assigned to the nearest BSs and BSs are connected to the nearest switching centers by cable, resulting in the network architecture in Fig.~\ref{Fig:Voronoi}.  In each time slot, a BS in an empty cell is silent and a BS in a nonempty cell transmits to  a single mobile selected from mobiles in the same cell with equal probability, which is called a typical active mobile. We assume that all BSs use the same transmission power $\mu$. Then the signal transmitted by a BS $Y$ is received at the intended mobile with the power $S_Y = \mu h_Y D^{-\alpha}_Y$ 
where $\{h_Y\}$ are i.i.d. $\text{exp}(1)$ random variable modeling Rayleigh fading, $D_Y$ is the transmission distance,  and $\alpha$ is the path-loss exponent.  By adopting a typical  model (see e.g., \cite{EnergyEfficiencyBSDeployment:Fred}),   the power $P_b$ consumed by a BS is given as 
$P_b = \mathcal{A}\mu + \mathcal{B}$ where $\mathcal{A}$ is a constant and $\mathcal{B}$ is the offset  power consumed regardless of if the BS is transmitting. 

\begin{figure}
\centering
\includegraphics[width=7.5cm]{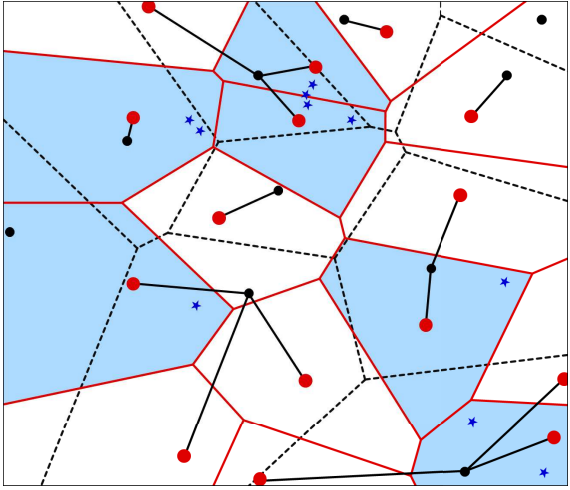}\vspace{5pt}
\caption{The stochastic-geometry model of a  cellular network. The base stations,  switching centers and mobiles are marked with big dots, little dots, and stars, respectively. Each BS is connected to the nearest switching center by  cable and serves a single mobile in the corresponding cell.  } \label{Fig:Voronoi}
\end{figure}

In this letter, we consider outage probability for a typical active mobile as a network performance metric. Let $U^\star$ and $Y^\star$ denote a typical active mobile and the corresponding serving BS, respectively. Assuming an interference-limited network, correct decoding at each mobile requires that the received signal-to-interference ratio (SIR) is above a given threshold $\theta$. Then outage probability $\Pout$ for  $U^\star$  is  
$\Pout = \mathrm{Pr}(\mathsf{SIR}< \theta).$ Note that since the BSs and mobiles are distributed as PPPs, the outage probability for a typical active mobile is independent with the number of other mobiles in the same cell.

\section{Mobile Performance} \label{Sec:Performance}

\subsection{Approximation of the Transmitting BS Process}
Even though the BSs are distributed as a PPP, the process of transmitting BSs is not because the empty-cell probability for each BS depends on the cell area that is determined by BSs' relative locations. For tractability, the transmitting BSs are modeled as a homogeneous PPP derived by thinning the BS process $\Sigma_b$ as follows. Let $p$ represent the empty-cell probability, namely the probability that  the typical BS is assigned no mobile. Mathematically, 
\begin{equation}
p = \Pr\left( |Y^\star - U| > |Y - U| \ \forall\ Y \in \Sigma_b\backslash\{Y^\star\},  U \in \Sigma_u \right). \nn
\end{equation}
We can obtain $p$ using the  following result from \cite{Ferenc:PoissonApprox} that approximates  the distribution function of $S$, the area of a typical Voronoi cell in a Poisson random tessellation
\begin{equation} \label{Approx:Voronoi}
f_S(x) \approx \frac{3.5^{3.5}}{\Gamma(3.5)} \lambda_b^{3.5} x^{2.5} e^{-3.5\lambda_b x}.
\end{equation}
Specifically, using \eqref{Approx:Voronoi} and the fact that $p$ is the average void probability of the mobile process for a typical Voronoi cell, 
\begin{align}
p & = \int_{0}^\infty e^{-\lambda_u x}f_S(x)dx \\
& \approx \left(1+\frac{\lambda_u}{3.5\lambda_b}\right)^{-3.5}\label{Eq:OffProb:a}\\
& \approx  1-\frac{\lambda_u}{\lambda_b}, \qquad \lambda_b/\lambda_u\rightarrow \infty.  \label{Eq:OffProb}  
\end{align}
It can be observed from~\eqref{Eq:OffProb} that $p$ increases with the growing ratio $\lambda_b/\lambda_u$. Intuitively, by increasing the BS density with the mobile density fixed, cells shrink and as a result the empty-cell probability grows. Given $p$, the transmitting BSs can be approximated as a PPP $\tilde{\Phi}_b$ with the density $(1-p)\lambda_b$. The results derived based on this  approximation is shown by simulation to be accurate as  $\lambda_b/\lambda_u\rightarrow \infty$. 

\subsection{Mobile Outage Probability}
 In terms of $p$, $\Pout$ for a typical active mobile $U^\star$ can be obtained as shown in the following proposition. 

\begin{proposition}\label{Prop:Pout} Given that  the transmitting BSs are approximated by $\tilde{\Phi}_b$,   the outage probability $\Pout$ for a typical active mobile $U^\star$ is  
\begin{equation} \label{Approx:Pout}
\Pout \approx \beta \frac{\lambda_u}{\lambda_b}, \qquad \lambda_b/\lambda_u\rightarrow \infty
\end{equation}
where  $\beta = \theta^{\frac{2}{\alpha}}\int_{\theta^{-\frac{2}{\alpha}}}^\infty \frac{1}{1 + x^{\frac{\alpha}{2}}} d x.$
\end{proposition}
\begin{proof}
The interference power $I$ for $U^\star$ is 
\begin{equation}\label{Eq:I:Dist}
I = \sum_{Y\in \tilde{\Phi}_b\backslash \{Y^\star\}} \mu h_Y |U^\star - Y|^{-\alpha} 
\end{equation}
which is known as a shot noise process. 
Moreover, define the transmission  distance $D^\star = |U^\star - Y^\star|$ and $D^\star$ has the following distribution function \cite{FossZuyev:VoronoiProcessPoisson:1996, Andrews:TractableApproachCoverageCellular:2010}
\begin{equation}\label{Eq:D:PDF}
f_{D^\star}(r) = 2\pi\lambda_b r e^{-\pi \lambda_b r^2}. 
\end{equation}
 Using \eqref{Eq:I:Dist}, \eqref{Eq:D:PDF} and the same procedure as \cite[$(15)$]{Andrews:TractableApproachCoverageCellular:2010} with some minor modifications, 
\begin{equation}
\Pout = 1-\frac{1}{1+(1-p)\beta} . \label{Eq:Pout}
\end{equation}
For $\lambda_b/\lambda_u\rightarrow \infty$, $\Pout$ can be approximated using \eqref{Eq:OffProb} as
\begin{align}
\Pout &\approx 1-\frac{1}{1+\beta \frac{\lambda_u}{\lambda_b}}, \qquad \lambda_b/\lambda_u\rightarrow \infty \label{Approx:Pout:a}\\
& = \beta\frac{\lambda_u}{\lambda_b} + O\l(\l(\frac{\lambda_u}{\lambda_b}\r)^2\r). \label{Approx:Pout:b}
\end{align}
The desired result follows from \eqref{Approx:Pout:b}.
\end{proof}

Despite the similarity in proof, the result in Proposition~\ref{Prop:Pout} focuses on the regime of many BSs while that in \cite[$(15)$]{Andrews:TractableApproachCoverageCellular:2010} assumes that mobiles significantly outnumber BSs. Consequently, the outage probability derived in \cite{Andrews:TractableApproachCoverageCellular:2010} is insensitive to the variation on the BS density $\lambda_b$ but $\Pout$ in  \eqref{Approx:Pout} decreases approximately linearly with increasing $\lambda_b$. This relation between $\Pout$ and $\lambda_b$ rises from two factors. First,  with $\lambda_u$ fixed, the density of transmitting BSs and hence the interference power measured at mobiles remain constant even as $\lambda_b$ increasing. Second, the received signal power at an active mobile increases as $\lambda_b/\lambda_u\rightarrow \infty$ due to the reduced transmitting distance.

\section{Optimization of The Base-Station Density} \label{Sec:OptBSDen}

In this section, we analyze the optimal tradeoff between the performance gain by employing many BSs and the resultant increased network cost in terms of BS energy consumption and installation,  and backhaul cable length. 

The BS density $\lambda_b$ can be optimized by minimizing the following multi-objective function $\mathcal{C}(\lambda_b)$: 
\begin{align}\label{Eq:CostFunc:Pout}
\mathcal{C}(\lambda_b)  = c_1 \mathcal{L} + c_2 \lambda_b + c_3 P_{\Sigma} + \phi \Pout
\end{align}
where $c_1$, $c_2$, $c_3$ denote the cost of laying a cable of a unit length, the hardware cost of an individual BS, and the price of consuming one-unit power, and $\phi$ represents the penalty for an event that a mobile is in outage. Moreover, $\mathcal{L}$ represents the average cable length per unit area and follows from  \cite{FossZuyev:VoronoiProcessPoisson:1996} as 
\begin{equation} \label{Eq:Cable}
\mathcal{L}  = \frac{\lambda_b}{2 \lambda_s^{3/2}}. 
\end{equation}
The total power $P_{\Sigma}$ consumed by all BSs can be written as
\begin{align} 
P_{\Sigma} &= (\mathcal{A}\mu + \mathcal{B})(1-p)\lambda_b + \mathcal{B}p\lambda_b  \nn\\
&\approx \mathcal{A}\mu\lambda_u + \mathcal{B}\lambda_b, \qquad \lambda_b/\lambda_u\rightarrow \infty
\label{Eq:PowerNetwork}
\end{align}
where \eqref{Eq:PowerNetwork} follows from \eqref{Eq:OffProb}. By substituting \eqref{Approx:Pout} and  \eqref{Eq:PowerNetwork} into \eqref{Eq:CostFunc:Pout} and minimizing the resultant cost function, the optimal BS density $\lambda_b^*$ is obtained as 
\begin{equation} \label{Eq:OptBS1}
\lambda_b^* = \sqrt{\frac{\phi}{\frac{c_1}{2\lambda_s^{3/2}} + c_2 +c_3 \mathcal{B}}\beta\lambda_u}. 
\end{equation}
It can be observed that  $\lambda_b^*$ is an increasing function of $\phi$ and $\lambda_u$. This suggests that it is desirable to install more BSs given more serve penalty for outage events or a larger number of mobiles. Combining the definition of $\beta$ in Proposition~\ref{Prop:Pout} and \eqref{Eq:OptBS1}, it can be inferred  that the optimal BS density should grow with the increasing target SIR $\theta$. Moreover, $\lambda_b^*$ grows as $\lambda_s$ increases   since the cable cost for installing an additional BS is reduced. In addition, $\lambda_b^*$ is also larger for lower average aggregate prices for one BS $\l(\frac{c_1}{2\lambda_s^{3/2}} + c_2 +c_3 \mathcal{B}\r)$, which agrees with intuition.

\section{Simulation} \label{Sec:Sim}

In this section, the derived  $\Pout$ and optimal BS density $\lambda_b^*$ are validated using simulation. In the simulation, the BSs are modeled as a homogeneous PPP in $\mathds{R}^2$.  The  simulation parameters are set as  $\theta = 3$ dB, $\alpha = 3$, and   $\lambda_u=0.02$. 

\begin{figure}
\centering
\includegraphics[width=9cm]{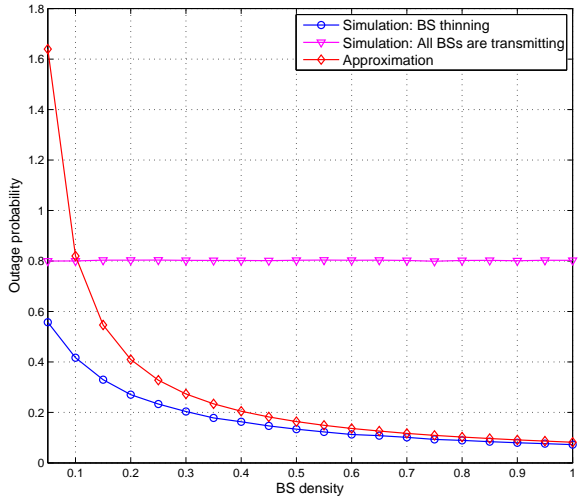}
  \caption{Comparison between the derived outage probability and simulation results for  increasing BS density $\lambda_b$ with fixed $\lambda_u = 0.02$} \label{Fig:Pout}
\end{figure}

\begin{figure}
\centering
\includegraphics[width=9cm]{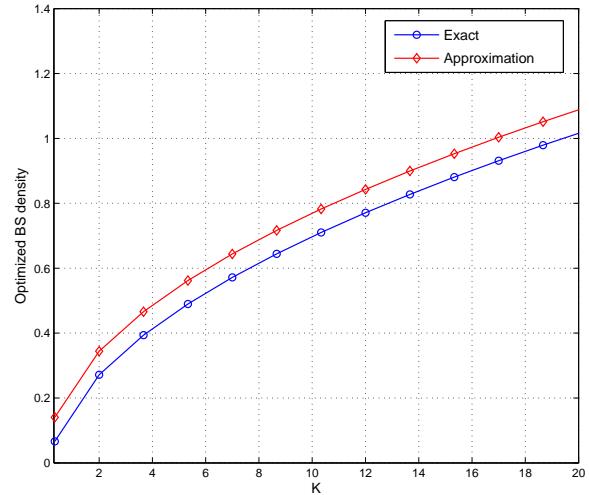}
\caption{Comparison between the optimal BS density ($\lambda_b^*)$ based on approximated $\Pout$ and exact value obtained numerically for increasing $K$ with fixed $\lambda_u = 0.02$.} \label{Fig:OptBSDen} 
\end{figure}

Fig.~\ref{Fig:Pout} compares the outage probability computed using \eqref{Approx:Pout} and its simulated values. They can be observed to converge as $\lambda_b$ increases. Also plotted in Fig.~\ref{Fig:Pout} is the outage probability given that all BSs transmit, which is insensitive to the changes on $\lambda_b$ as also observed in  \cite{Andrews:TractableApproachCoverageCellular:2010}. 

The accuracy of $\lambda_b^*$ is verified in Fig.~\ref{Fig:OptBSDen}. Using the factors in \eqref{Eq:OptBS1}, let us define a parameter 
\begin{equation}
K = \frac{\phi}{\frac{c_1}{2\lambda_s^{3/2}} + c_2 +c_3 \mathcal{B}}
\end{equation}
which represents the ratio of the penalty for outage to average aggregate cost for an one BS. Fig.~\ref{Fig:OptBSDen} compares $\lambda_b^*$ and the optimal density obtained numerically by varying $K$. It can be observed that by increasing $K$, $\lambda_b^*$ traces the exact value within a constant and small gap ($\approx 0.07$). With the gap unchanged, this result reveals that the percentage of error for the derived density reduces with increasing $K$. The error is shown to be lower than $10\%$ with $K>10$. It means that $\lambda_b^*$ is accurate when the penalty for outage is larger than the price for operating each BS.

\section{Conclusion} \label{Sec:Con}
In this letter, we have quantified the performance gains by employing many BSs in cellular networks. Specifically, it has been shown that the outage probability decreases inversely with the increasing BS density. Moreover, the BS density has been optimized by considering both performance gains and network cost including BS hardware, energy consumption and backhaul cables. The optimal BS density has been shown to increase sub-linearly with the growing mobile density given the performance metric of outage probability. Moreover, larger  BS density is desirable for lower prices for power, cables and BS hardware. 

\section*{Acknowledgement}
The authors thank Seung Min Yu for helpful discussions.

\bibliographystyle{ieeetr}

\end{document}